\documentclass[letterpaper, 10 pt, conference]{ieeeconf}  %

\IEEEoverridecommandlockouts                              %

\usepackage{amsmath,amssymb,amstext}
\usepackage{physics}
\usepackage{xcolor}

\colorlet{lcolor}{blue!40!black}
\colorlet{ucolor}{magenta!40!black}
\colorlet{ccolor}{green!40!black}

\usepackage[colorlinks=true,%
linkcolor=lcolor,%
urlcolor=ucolor,%
citecolor=ccolor]{hyperref}

\usepackage{algorithm}
\usepackage{algorithmic}

\usepackage{pgfplots} %
\pgfplotsset{compat=newest} 
\pgfplotsset{plot coordinates/math parser=false}

\usepackage{mathtools}

\DeclarePairedDelimiter\floor{\lfloor}{\rfloor}

\usepackage{tikz}
\usepackage{circuitikz}
\usetikzlibrary{positioning}
\usetikzlibrary{backgrounds}
\usetikzlibrary{shapes,arrows,chains} 

\newtheorem{asum}{Assumption}
\newtheorem{theo}{Theorem}
\newtheorem{rema}{Remark}

\newtheorem{lemma}{Lemma}
\newtheorem{prop}{Proposition}
\newtheorem{coro}{Corollary}

 \newcommand{\conc}[2]{\left[{#1}\right]_{{#2}}}

 \newcommand{\CLF}{\gamma}
 
 \newcommand{\boundd}{{b}}

 \newcommand{\partder}[1]{V_\CLF^{'}(#1) }
 \newcommand{\timeder}[2]{\mathcal{L}_fV_\CLF(#1,#2)}
 \newcommand{\tone}{r}
 \newcommand{\ttwo}{m}
 \newcommand{\modelbased}{model-based }
  \newcommand{\modelbasedPETC}{MB-PETC }
  \newcommand{\remove}[1]{}

\title{\LARGE \bf
	Model-Based Nonlinear  Periodic Event-Triggered Control for Continuous-Time Systems with Sampled-Data Prediction
}

\author{Michael Hertneck, Steffen Linsenmayer and Frank Allg\"ower%
\thanks{The authors thank the German Research Foundation (DFG) for support of this work within grant AL 316/13-2  and within the German Excellence Strategy under grant EXC-2075.}%
\thanks{The authors are with the Institute for Systems Theory and Automatic Control, University of Stuttgart, 70569 Stuttgart, Germany (email: $\{$hertneck, linsenmayer,  allgower$\}$@ist.uni-stuttgart.de).}
}

\renewenvironment{thebibliography}[1]{%
	\begin{oldthebibliography}{#1}%
		\setlength{\parskip}{0ex}%
		\setlength{\itemsep}{0ex}%
	}%
	{%
	\end{oldthebibliography}%
}
\begin{document}

\pubid{\begin{minipage}{\textwidth}\ \\[12pt] This version has been accepted for publication in Proc. European Control Conference (ECC), 2020.  Personal use of this material is permitted. Permission from EUCA must be obtained for all other uses, in any current or future media, including reprinting/republishing this material for advertising or promotional purposes, creating new collective works, for resale or redistribution to servers or lists, or reuse of any copyrighted component of this work in other works.\end{minipage}}

\maketitle
\begin{abstract}
	
In this paper, we present a model-based periodic event-triggered control mechanism for nonlinear continuous-time Networked Control Systems. A sampled-data prediction of the system behavior is used at the actuator to reduce the amount of required communication while maintaining a user-defined performance level. This prediction is based on a possibly inaccurate discretization of the nonlinear system dynamics and can be implemented on simple hardware. Nevertheless, guarantees for asymptotic stability and a user-defined performance level are given for the periodic event-triggered control (PETC) mechanism, whilst the reduction of the required amount of transmissions of state information depends on the quality of the prediction.  We discuss furthermore how the prediction can be implemented.
 The performance of the proposed PETC mechanism is illustrated with a numerical example. This paper is the accepted version of \cite{hertneck20ecc}, containing also the proofs of the main results.
	
\end{abstract}

\linespread{0.97}
\section{Introduction}
A major challenge in the field of Networked Control Systems (NCS) is to find sampling and control strategies that require only a small amount of communication while still guaranteeing stability and a certain degree of performance for the control system \cite{Hespanha2007}. By prolongating the time between transmissions, a reduction of the  required amount of communication for the NCS can be achieved.  However, this comes often at the cost of decreasing performance.  
In literature, there are different sophisticated techniques to extend the time span between transmissions, while still obtaining a good trade off between these two potentially conflicting objectives. 
\pubidadjcol

A widespread approach to achieve such a good trade off, is event-triggered control (ETC) \cite{Heemels2012,tabuada2007event,ong2018event,proskurnikov2018lyapunov}. In ETC, control updates are not send periodically, but instead according to a state dependent trigger rule. This stands in contrast to time-triggered control, where control updates are triggered periodically with a fixed sampling period. Since it is impossible to evaluate the state dependent trigger rule continuously on digital hardware, periodic event triggered control (PETC) \cite{proskurnikov2018lyapunov,heemels2013periodic,Wang2016} has been developed. Here, the trigger rule is only evaluated periodically at fixed sampling times, but nevertheless stability can be guaranteed. 
\remove{Nevertheless, stability and a certain performance level can be guaranteed, whilst the required amount of communication can be reduced compared to time-triggered control.}

Another technique that is capable of reducing the amount of communication, while still guaranteeing a certain performance level, is based on using a model at the actuator to predict the plant behavior. Then, the system input is calculated based on this prediction, if no current state information is available \cite{montestruque2003model,polushin2008model,mastellone2005model}. We will refer to this technique subsequently as model-based networked control (MBNC).  

To exploit the advantages of both PETC and MBNC, it is hence desirable to combine both techniques. For ETC, MBNC and linear systems, this has been done successfully in \cite{lunze2010state,garcia2013model}. A linear model-based PETC framework has been presented in \cite{Heemels2013}. Whilst \modelbased PETC (MB-PETC) for linear systems is thus well understood, most nonlinear PETC approaches from literature do not take into account a prediction at the actuator. 
\remove{Thus, there are only few nonlinear \modelbasedPETC approaches available. }
For discrete-time systems that satisfy a dissipation inequality, nonlinear \modelbasedPETC was investigated in \cite{mcCourt2014model}. Moreover, the nonlinear PETC approach from \cite{Borgers2018a} can in principle be modified to include a prediction with a continuous-time prediction model at the actuator, even though \cite{Borgers2018a} focuses on the prediction-free case. Since the research for nonlinear MB-PETC is hence still at an early stage, it is desirable to investigate \modelbasedPETC mechanisms. 

In this paper, we present an \modelbasedPETC mechanism for nonlinear systems that is \pubidadjcol 
based on the linear PETC approach from \cite{linsenmayer2018event2}. The proposed mechanism uses a possibly inaccurate discretization of the considered plant as sampled data prediction model. The prediction of the system state, that is used to determine the system input at the actuator if no current state information is available,  thus only needs to be updated periodically at sampling times of the PETC mechanism. This stands in contrast to \cite{Borgers2018a}, where a continuous-time prediction model would be required, that cannot be implemented on digital hardware. The sampled-data prediction model, on the contrary makes the proposed approach feasible, even if the actuator has only very limited computational capabilities. 
 Stability and a chosen convergence speed are guaranteed independent of the actual choice of the prediction model.
  To achieve this, a copy of the prediction model is used by the PETC mechanism at the actuator side to include knowledge about the next input that will be applied to the system for the trigger decision. The reduction of the required amount of communication depends on the quality of the prediction model. We discuss how a suitable prediction model can be obtained.
  Finally, we demonstrate with a numerical example that the proposed approach can improve the system performance and  nevertheless significantly reduce the required amount of communication in comparison to the prediction-free PETC mechanism from \cite{hertneck2019nonlinear}.

The remainder of this paper is structured as follows. In Section~\ref{sec_setup}, we introduce the considered setup and specify the control objective. Then, we discuss how the sampled-data prediction can be implemented on simple hardware in Section~\ref{sec_prediction}. Some basic results for the PETC design are presented in Section~\ref{sec_basic}. Building up on this, we propose the \modelbasedPETC mechanism in Section~\ref{sec_main}. A conclusion in Section~\ref{sec_conc} completes the paper. Some spacious proofs are given in the appendix.

\subsubsection*{Notation}
The positive (respectively nonnegative) real numbers are denoted by $\mathbb{R}_{>0}$,  respectively  $\mathbb{R}_{\geq 0} := \mathbb{R}_{> 0} \cup \{0\} $. The positive (respectively nonnegative) natural numbers are denoted by $\mathbb{N}$, respectively $\mathbb{N}_0:=\mathbb{N}\cup  \left\lbrace 0 \right\rbrace $. $\floor{x}$ returns the next integer to $x$ that is smaller than or equal to $x$. A continuous function $\alpha: \mathbb{R}_{\geq 0} \rightarrow \mathbb{R}_{\geq 0}$ is a class $ \mathcal{K}$ function if it is strictly increasing and $\alpha(0) = 0$. The notation $t^-$ is used as $t^- := \lim\limits_{s<t,s\rightarrow t} s$. A continuous function $V:\mathbb{R}^n \rightarrow \mathbb{R}$ is positive definite if $V(0) = 0$ and $V(x)>0$ for all  $x\neq 0$.  $V'(k)$ denotes $\left. \frac{\partial V(x)}{\partial x} \right|_{x = k}$.  Furthermore, we use in a slight abuse of notation $\mathcal{L}_fV(x,u)$ to denote the Lie derivative of $V$ along the vector field  $f:\mathbb{R}^n \times \mathbb{R}^m \rightarrow \mathbb{R}^n$, i.e., $\mathcal{L}_fV(x,u) = V^{'}(x)  f(x,u) $.
We use   $\conc{f}{j}(\cdot) = \underbrace{f\circ\dots\circ f (\cdot)}_{j-times}$ as operator for $j$ concatenations of a function $f$. 
\section{Setup}
\label{sec_setup}
In this section, we specify the setup for this paper, including the considered control system, the \modelbasedPETC strategy and the control objective.
\subsection{Control System}
The setup that we consider is sketched in Figure~\ref{fig_setup}. The plant to be controlled is described by a nonlinear, time-invariant system
\begin{equation}
\label{eq_sys_cont}
\dot{x} = f(x,u)
\end{equation}
with a vector valued function $f:\mathbb{R}^{n_x}\times\mathbb{R}^{n_u} \rightarrow \mathbb{R}^{n_x}$ satisfying $f(0,0) = 0$, %
the system state $ x(t)\in \mathbb{R}^{n_x}$ with initial condition $x(0) = x_0$ and the input $u(t)\in\mathbb{R}^{n_u}$. The input is generated by a controller as
\begin{equation}
\label{eq_cont_fb}
u = \kappa(\hat{x})
\end{equation} 
with the nonlinear feedback law $\kappa:\mathbb{R}^{n_x}\rightarrow\mathbb{R}^{n_u}$ and a prediction $\hat{x}(t)$ of the system state $x(t)$ that is generated at the actuator based on received state information and a prediction mechanism.

The system state is sampled uniformly with a fixed sampling period $h \in \mathbb{R}_{>0}$ that will be specified later. Thus, at each discrete time instant satisfying $t=kh$ for some $k\in\mathbb{N}_0$, a new sample of the current system state is drawn. We define the sequence of sampling times as $(\tau^s_k)_{k\in\mathbb{N}_0}$, i.e., $\tau^s_k = kh$ and introduce the infinite set of sampling times as $\mathcal{T}^s := \left\lbrace \tau^s_0,\tau^s_1,\tau^s_2,\dots \right\rbrace$. However, to reduce the network load, current state information is not transmitted at each sampling time  to the actuator. Instead, the time instants, when state information is transmitted to the actuator are given by the infinite sequence $(\tau^a_l)_{l \in \mathbb{N}_0}$  and define a discrete set 
$\mathcal{T}^a := \left\lbrace \tau^a_0, \tau^a_1,\tau^a_2,\dots \right\rbrace \subseteq \mathcal{T}^s.$
The set $\mathcal{T}^a$ depends on a PETC mechanism that will be designed in this paper. The PETC mechanism will guarantee that $\tau^a_0 = \tau^s_0 = 0$.

The update of the predicted state $\hat{x}$ when current state information is received at the actuator is represented by $\hat{x}(\tau^a_l)  = x(\tau^a_l)$ for all $\tau^a_l \in \mathcal{T}^a$.  Between the update times, a state prediction is used at the actuator. We assume however, that a continuous-time prediction model cannot be implemented at the actuator. This is for example the case for digital hardware. Instead, in this paper, a sampled-data approach is used for the prediction. This approach consists of updating $\hat{x}$ at each sampling time when no state information is received according to 
\begin{equation}
\label{eq_cont_est}
\hat{x}(\tau^s_k) = f_p(\hat{x}(\tau^s_{k-1}))~ 
\end{equation} 
if $\tau^s_k \in \mathcal{T}^s \backslash \mathcal{T}^a$,
with a continuously differentiable prediction function $f_p$ which can be chosen arbitrarily as long as $f_p(0) = 0$.
Between sampling times, $\hat{x}$ is kept constant, i.e., $\dot{\hat{x}}(t) = 0$ for all $ t \notin \mathcal{T}^s$. This allows the implementation of the $\hat{x}$ dynamics even on simple hardware.
Zero-order hold (ZOH) feedback is a special case of the prediction with $f_p(\hat{x}) = \hat{x}$ and resembles the setup of \cite{hertneck2019nonlinear}.

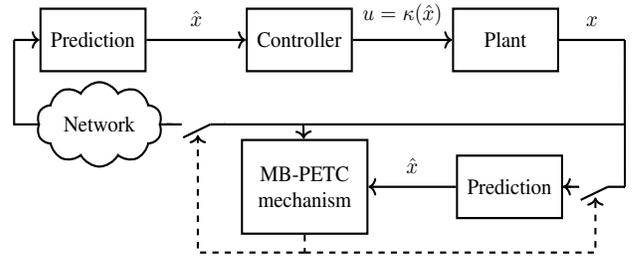
\begin{figure}
	\resizebox{\linewidth}{3.5cm}{
	\begin{circuitikz}[line width=1.0]
		\usetikzlibrary{calc}
		\ctikzset{bipoles/thickness=1}
		\node[draw,minimum width=1.8cm,minimum height=1cm,anchor=south west] at (0,0) (node1){Plant};
		\node[draw,left = 17mm of node1, minimum width=1.8cm,minimum height=1cm] (node3) {Controller};
		\node[draw,left = 17mm of node3, minimum width=1.8cm,minimum height=1cm] (node4) {Prediction};

		\node[below = 17	mm of node3] (helpnode13) {};
		\coordinate (Q13) at (helpnode13);
		\node[draw,right = -10mm of Q13, minimum width=2.15cm,minimum height=1.5cm,align = center] (node5) {\modelbasedPETC\\ mechanism};
		\node[draw,right = 15mm of node5, minimum width=1.8cm,minimum height=1cm,align = center] (node6) {Prediction};
		
		\node[below = 7 mm of node3] (helpnode0) {};
		\node[left = 10 mm of helpnode0] (helpnode1) {};
		\node[left = 3 mm of node4] (helpnode2) {};
		\node[left = 5 mm of helpnode1] (helpnode3) {};
		\node[below = 1.5 mm of node5] (helpnode4) {};
		\node[left = 2 mm of helpnode1] (helpnode5) {};
		\node[below = -1 mm of helpnode5] (helpnode6) {};
		\node[right = 9.9mm of node1] (helpnode7) {};
		\node[right = 9.25mm of node6] (helpnode8) {};
		\node[right = -1.5mm of node6] (helpnode9) {};
		\node[right = 3 mm of helpnode9] (helpnode10) {};
		\node[below = -1 mm of helpnode10] (helpnode11) {};
		\node[right = 0 mm of node6] (helpnode12) {};
		
		\node[right = -2 mm of helpnode0] (helpnode14) {};

			\node [cloud, draw, left =2mm of helpnode3, cloud puffs=10,cloud puff arc=120, aspect=2, inner ysep=0.5em](cloud) {Network};

		\coordinate (Q0) at (helpnode0);
		\coordinate (Q1) at (helpnode1);
		\coordinate (Q2) at (helpnode2);
		\coordinate (Q4) at (helpnode4);
		\coordinate (Q6) at (helpnode6);
		\coordinate (Q7) at (helpnode7);
		\coordinate (Q8) at (helpnode8);
		\coordinate (Q9) at (helpnode9);
		\coordinate (Q11) at (helpnode11);
		\coordinate (Q12) at (helpnode12);
		\coordinate (Q14) at (helpnode14);

		 \draw (Q1) to[normal open switch] ([xshift=-1cm]helpnode1);
		 \draw (Q8) to[normal open switch] (Q9);
		 
		 \draw [-,line width = 1pt] (helpnode3) -- (cloud);
		 \draw [-,line width = 1pt] (cloud) -| (Q2);
		  \draw [->,line width = 1pt] (Q2) -- (node4);
		  \draw [-,line width = 1pt,dashed] (node5) -- (Q4);
		  \draw[->,line width = 1pt,dashed] (Q4) -| (Q6);
		  \draw [-,line width = 1pt] (node1) -- node [text width=1.5cm,midway,above = 0.25em,align = center] {$x$} (Q7);
		  \draw [-,line width = 1pt] (Q7) |- (Q1);
		  \draw [->,line width = 1pt] (Q14) -- (node5);
		  \draw [-,line width = 1pt] (Q7) -- (Q8);
		  \draw[->,line width = 1pt,dashed] (Q4) -| (Q11);
		  \draw [->,line width = 1pt] (Q12) -- (node6);
		  
		  \draw [->,line width = 1pt] (node4) -- node [text width=1.5cm,midway,above = 0.25em,align = center] {$\hat{x}$} (node3);
		  \draw [->,line width = 1pt] (node3) -- node [text width=1.5cm,midway,above = 0.25em,align = center] {$u = \kappa(\hat{x})$} (node1);
		  \draw [->,line width = 1pt] (node6) -- node [text width=1.5cm,midway,above = 0.25em,align = center] {$\hat{x}$} (node5);

	\end{circuitikz}
}
	\vspace{-5mm}
	\caption{Sketch of the considered setup.}
	\label{fig_setup}
	\vspace{-7mm}
\end{figure}

The closed-loop system consists of system \eqref{eq_sys_cont}, controller \eqref{eq_cont_fb}, prediction \eqref{eq_cont_est}
and its reset condition and can be modeled as a discontinuous dynamical system (DDS) with state $\xi = \left[\xi_1^\top, \xi_2^\top \right]^\top =  \left[x^\top, \hat{x}^\top\right]^\top$  for $k \in \mathbb{N}_0$ as 
\begin{align}
\begin{split}
\dot{\xi}(t) =& 
\begin{bmatrix}
f(\xi_1(t),\kappa(\xi_2(t))) \\
0
\end{bmatrix},~\tau^s_{k} \leq t < \tau^s_{k+1}
, \\
\xi(t) =& \begin{bmatrix}
\xi_1(t^-)\\
f_p(\xi_2(\tau^s_k))
\end{bmatrix}, t  = \tau^s_{k+1} \wedge t \notin  \mathcal{T}^a 
\\
\xi(t) =& \begin{bmatrix}
 \xi_1(t^-)\\
 \xi_1(t^-)
\end{bmatrix},  t  = \tau^s_{k+1} \wedge t \in  \mathcal{T}^a 
\label{eq_cont_sys_comp}
\end{split}
\end{align}
 with $\xi(0) = \xi_0 := \left[x_0^\top,\hat{x}_0^\top\right]^\top$ with $\hat{x}_0 = x_0$ due to $\tau^a_0 = 0$.

\subsection{PETC Triggering Strategy}
The time instants when the system state is transmitted to the actuator are determined using PETC, i.e., according to a trigger mechanism that is evaluated at each sampling time, and thus at each element of $\mathcal{T}^s$. As depicted in Figure~\ref{fig_setup}, the PETC trigger mechanism has access to a copy of the prediction that is used at the actuator. Using the feedback law $\kappa(\hat{x})$, it can thus also include knowledge about the input that will be applied to the plant during the next sampling period in the trigger decision. If the trigger rule of the PETC mechanism is violated at a sampling time, then the current system state is transmitted to the actuator.

In order to design the PETC mechanism, we assume that a stabilizing controller for the continuous-time system has already been synthesized. Thus, a continuous-time feedback law and a Lyapunov function are known, that satisfy the following assumption.
\begin{asum}
	\label{asum_cont_clf}
	There is a continuous, positive definite function $V_\CLF:\mathbb{R}^n \rightarrow \mathbb{R}$, satisfying	
	\begin{equation}
	\label{eq_cont_clf2}
	\alpha_1 {(\norm{x})} \leq V_\CLF(x) \leq \alpha_2 {(\norm{x})} 	
	\end{equation}	
	\begin{equation}
	\label{eq_cont_clf}
	V_\CLF^{'}(x(t))  f(x(t),\kappa(x(t))) \leq - \gamma(V_\CLF(x(t)))
	\end{equation}
	with class $\mathcal{K}$ functions $\gamma, \alpha_1, \alpha_2 $.
\end{asum}
Finding $\kappa$ and $V_\CLF$ that satisfy Assumption~\ref{asum_cont_clf} is a common but non-trivial problem in control theory for continuous-time systems and is widely discussed in literature, see e.g.  \cite{khalil2002nonlinear}. 

Subsequently, we consider local results for a level set of $V_\CLF$, that can be defined as $\mathcal{X}_c := \left\lbrace x | V_\CLF(x) \leq c \right\rbrace$ for a chosen $c \in 
\mathbb{R}_{>0}$. Next, we present a criterion on the convergence speed of the system state, which will be used in this paper as performance measure, and define the control objective that has to be guaranteed by the \modelbasedPETC mechanism.

\subsection{Convergence Criterion and Control Objective}
We consider in this paper the same averaged convergence criterion that was used in \cite{ong2018event} for event-triggered control with performance barrier, except for a time shift of one sampling period. Let $S(t,x_0)$ be the solution to 
\begin{equation}
\label{eq_cont_def_s}
\frac{d}{dt} S(t,x_0) = -\sigma \gamma(S(t,x_0)), ~S(0,x_0) = V_\CLF(x_0).
\end{equation}
for a chosen $\sigma \in 
\left(0,1\right)$. We require as performance criterion %
\begin{equation}
\label{eq_cont_v_smaller}
V_\CLF(x(t+h)) \leq S(t,x_0)
\end{equation}
for all $t>0$. This means, that we require in average the same convergence speed as can be guaranteed for  continuous-time feedback except for a tunable deviation, that is described by $\sigma$, and a time shift of one sampling period. Note that this time shift is small as long as $h$ is small. We denote subsequently the maximum admissible sampling period such that \eqref{eq_cont_v_smaller} can be guaranteed  as $\sigma$-MASP.
The following proposition, that is taken from \cite{hertneck2019nonlinear} can be used to guarantee satisfaction of \eqref{eq_cont_v_smaller} without explicit knowledge $S(t,x_0)$. \remove{, and will turn out to be useful in the PETC design.}

\begin{prop}[Proposition~3 from \cite{hertneck2019nonlinear}]
	\label{prop_cont_conv}
	Consider two constants $C_1$, $C_2 ~\in \mathbb{R}_{\geq 0}$,  and $S(t,x_0)$ defined by \eqref{eq_cont_def_s}. If 
$C_1   	\leq C_2- \tone \sigma  \gamma(C_2) \nonumber$
	and  $C_2 \leq S(s,x_0)$  	holds for $  \tone,s \in \mathbb{R}_{\geq 0}$,
	then 
	$C_1 \leq S(s+\tone,x_0).$ %
\end{prop}

The goal of this paper is to present an \modelbasedPETC mechanism, that exploits the prediction at the actuator to improve the system performance in comparison to existing approaches without prediction, whilst reducing the amount of transmissions of the system state. Nevertheless, asymptotic stability of the origin of the closed-loop system that is represented by the DDS~\eqref{eq_cont_sys_comp} and satisfaction of the convergence criterion~\eqref{eq_cont_v_smaller} have to be guaranteed for all initial conditions from the level set $\mathcal{X}_c$, even in case the prediction function is very inaccurate due to computational limitations of the actuator. 

\section{Sampled-Data Prediction for a Continuous-Time Plant} 
\label{sec_prediction}
In this section, we discuss options for choosing the prediction function $f_p$.
The choice of $f_p$ does significantly influence the closed-loop performance and the required amount of transmissions, and is thus important for the \modelbasedPETC. Choosing $f_p(x) = x$, which resembles the ZOH case and which is done for most nonlinear PETC mechanisms from literature, will in general not deliver a good prediction.  
Instead, a small amount of transmissions can e.g. be achieved, if  $f_p$ approximates the exact discretization of the plant dynamics \eqref{eq_sys_cont} with a small enough approximation error since then, the difference between $ x(t)$ and $\hat{x}(t)$ grows only slowly. If the exact discretization of $f$ could be used, one single transmission would even suffice to stabilize the plant if no additional perturbations occur. However, finding the exact discretization is impossible for most nonlinear systems. Moreover, the computational capabilities of the actuator may be limited, making the evaluation of an exact discretization impossible.
We present now different options how $f_p$ can be chosen, depending on the actuator.
\subsubsection*{Runge-Kutta discretization}
The prediction function $f_p$ can be chosen as a Runge-Kutta discretization of $f$. This works well if the actuator has suitable computational capabilities to evaluate the nonlinear function $f$ at any point in the state space.
\subsubsection*{Look-up tables}	
The prediction function $f_p$ can be constructed from a look-up table that is based on sampled values of a precise but computationally demanding discretization of $f$. Continuous differentiability of $f_p$ can then be ensured by interpolation between the samples. This works well, even if the actuator has only very limited computational capabilities.
\subsubsection*{Regression of a precise discretization}
A regression method as e.g. a neural network can be used to approximate a precise but computationally demanding discretization of $f$. Such regressions can be implemented on simple hardware with limited computational capabilities and  nevertheless often allow a good approximation. 
\remove{\subsubsection*{Sampling a continuous model}
If a continuous-time model according to  \eqref{eq_sys_cont} can be simulated at the actuator, e.g. in hardware,  then the prediction can be based on samples of this continuous-time model. }
\subsubsection*{Transmitting sequences of future inputs}
For packet based networks, a sequence of future updates of $\hat{x}$ can be computed by the trigger mechanism based on a precise but computationally demanding discretization of $f$, and an input sequence can be transmitted to the actuator. Then no computational capabilities are required at the actuator at all. In turn, the required packet size may be larger for that approach as for the other approaches presented here.

\begin{rema}
	Due to limited computational capabilities of the actuator, it may be impossible to compute $\kappa$ exactly at the actuator. Instead, the above discussed approaches can also be used to approximate $\kappa$. This approximation can be easily included in the proposed \modelbasedPETC mechanism. To obtain still guarantees, it is then only required to transmit an exactly computed input together with the current system state over the network  to the actuator at transmissions times. 
\end{rema}

\section{PETC Preliminaries}
\label{sec_basic}
In this section, we first recap some results from \cite{proskurnikov2018lyapunov} and \cite{hertneck2019nonlinear}. Then, we extend a sufficient stability condition for DDS based on non-monotonic Lyapunov functions from \cite{michel2015stability} to the setup of this paper. Both parts of this section will be used later to design the PETC mechanism.

\subsection{Preliminary Results from \cite{proskurnikov2018lyapunov} and \cite{hertneck2019nonlinear}}
First, we recap from \cite{proskurnikov2018lyapunov} how a time dependent and state independent upper bound on the time derivative of $V_\CLF(x(t))$ can be computed. This bound is used in \cite{proskurnikov2018lyapunov} and \cite{hertneck2019nonlinear}  to design PETC mechanisms and to  compute a lower bound on the $\sigma$-MASP for which stability and the convergence criterion  \eqref{eq_cont_v_smaller} can be guaranteed. We will use the bound on the time derivative of $V_\CLF(x(t))$ for the design of the \modelbasedPETC mechanism. 

First, we need some assumptions that specify the considered setup. A discussion of these assumptions is omitted here due to spacial limitations and can be found in \cite{proskurnikov2018lyapunov} and \cite{hertneck2019nonlinear}.

\begin{asum}
	\label{as_cont_pro1}
	(cf.  Assumption~1 and 2 from \cite{proskurnikov2018lyapunov}) For the chosen $c\in\mathbb{R}_{>0}$, there is a finite Lipschitz constant $L_{1,c}$ satisfying
$	L_{1,c} \overset{\triangle}{=} \underset{x_1,x_2,x_3\in\mathcal{X}_c, x_1 \neq x_2}{\sup} \frac{\norm{f(x_1,~\kappa(x_3))-f(x_2,~\kappa(x_3))}}{\norm{x_1-x_2}}.$
\end{asum}

\begin{asum}
	\label{as_cont_pro2}
	(cf.	Assumption~3 in \cite{proskurnikov2018lyapunov})   For the chosen  $c\in\mathbb{R}_{>0}$, there is a finite Lipschitz constant $L_{2,c} \in \mathbb{R}$ satisfying
	$L_{2,c} \overset{\triangle}{=} \underset{x_1,x_2\in\mathcal{X}_c, x_1 \neq x_2}{\sup} \frac{\norm{V_\CLF^{'}(x_1) -V_\CLF^{'}(x_2) }}{\norm{x_1-x_2}}.$
\end{asum}

\begin{asum}
	\label{as_cont_pro3}
	(cf. Assumption~4, resp. Lemma~1 from \cite{proskurnikov2018lyapunov}) For the chosen $c\in\mathbb{R}_{>0}$, there is a positive definite function $M_c:\mathbb{R}^n \rightarrow \mathbb{R},$ that is bounded on $\mathcal{X}_c$, satisfying for all $x\in\mathcal{X}_c$
	\begin{align*}
	&\norm{V_\CLF^{'}(x) } \norm{f(x,~\kappa(x))} + \norm{f(x,~\kappa(x))}^2\\
	\leq& M_c(x) \abs{\partder{x} f(x,~\kappa(x))}.
	\end{align*}
\end{asum}

In order to obtain a time dependent but state independent bound on $V_\CLF(x(t))$ for a constant input, we consider the additional Cauchy problem $\dot{\tilde{x}}(t) = f(\tilde{x},u^*), \tilde{x}(0) = \tilde{x}_0 \in \mathcal{X}_c$  for the chosen $c \in \mathbb{R}_{>0}$ and some constant $u^*$. We define $t^*$ as the first time after $t = 0$ for which $V_\CLF(\tilde{x}(t)) \geq c$  and $\triangle_*(\tilde{x}_0,u^*) = \left[0, t^*\right]$. Note that a unique solution to the Cauchy problem exists on $\left[0,t^*\right]$ due to Assumption~\ref{as_cont_pro1}. The following upper bound on the time derivative of $V_\CLF(\tilde{x}(t))$ has been derived in \cite{proskurnikov2018lyapunov}.

\begin{coro}
	\label{coro_cont_lyap}
	(deviated from Corollary~4 from \cite{proskurnikov2018lyapunov}, cf. Corollary 1 from \cite{hertneck2019nonlinear})	Let Assumptions~\ref{as_cont_pro1} and \ref{as_cont_pro2} hold for the  chosen $c\in\mathbb{R}_{>0}$. Let $\tilde{x}_0,~\tilde{x}_1 \in \mathcal{X}_c$ , $~u^* = \kappa(\tilde{x}_1)$, $t\in \triangle_*(\tilde{x}_0,u^*) \cap \left[0, (1+2L_{1,c})^{-1}\right]$ and $\tilde{x}(t)$ be the solution of $\dot{\tilde{x}}(t) = f(\tilde{x}(t),u^*), \tilde{x}(0) = \tilde{x}_0$. Then, 
	\begin{align}
	&\abs{\timeder{\tilde{x}(t)}{u^*} - \timeder{\tilde{x}_0}{u^*}} \nonumber \\
	\leq& \sqrt{t} \mu_c \left(\norm{ V_\CLF^{'}(\tilde{x}_0)}  \norm{f(\tilde{x}_0,u^*)} + \norm{f(\tilde{x}_0,u^*)}^2 \right), \nonumber
	\end{align}
	where $	\mu_c \triangleq \sqrt{e} \max \left\lbrace L_{1,c},L_{2,c}(1+L_{1,c} \sqrt{e})\right\rbrace.$
\end{coro} 	
Based on the upper bound on the time derivative of $V_\CLF(\tilde{x}(t))$ from Corollary~\ref{coro_cont_lyap}, there are different approaches to determine an upper bound on the $\sigma$-MASP. A bound that takes directly into account the time derivative of $V_\CLF(\tilde{x}(t))$ and that guarantees a certain amount of decrease of $V_\CLF(\tilde{x}(t))$ at any time has been presented in \cite{proskurnikov2018lyapunov}. The averaged convergence criterion \eqref{eq_cont_v_smaller} allows an increase of the Lyapunov function for some times if an average decrease is still guaranteed. Thus, we use in this paper instead the bound on the $\sigma$-MASP from \cite{hertneck2019nonlinear}, that is less conservative for the considered setup. This bound is given for $x_0 \in \mathcal{X}_c$ by 
	\begin{equation}
	\label{eq_h_max}
	h_{\sigma\text{\normalfont -MASP}} = \min \left\lbrace \left(\frac{3 (1-\sigma)}{2 \mu_c M_{\max,c}}\right)^2, (1+2L_{1,c})^{-1}\right\rbrace,
	\end{equation}
	where 	$M_{\max,c} = \underset{x \in \mathcal{X}_c}{\sup} M_c(x)$
	and $\sigma \in \left(0,1\right)$.
	If $h\leq h_{\sigma\text{\normalfont -MASP}}$, then we know from \cite[Remark~1]{hertneck2019nonlinear} that asymptotic stability and satisfaction of the convergence criterion \eqref{eq_cont_v_smaller} are guaranteed for periodic sampling  with sampling period $h$. 
	\remove{We will exploit this later for the PETC design.}

\subsection{Non-Monotonic Stability for the \modelbasedPETC Setup} 
Now, we present a sufficient stability condition for the DDS~\eqref{eq_cont_sys_comp} that is closely related to Theorem~6.4.2 from \cite{michel2015stability}. The main difference to Theorem~6.4.2 is that our DDS model~\eqref{eq_cont_sys_comp} has two different jump equations for $t\in \mathcal{T}^a$ and $t \in \mathcal{T}^s \backslash \mathcal{T}^a$. However, a decrease of the Lyapunov function can only be guaranteed along successive time instants from $\mathcal{T}^a$. A modified version of Theorem 6.4.2 that takes into account the additional jumps of the DDS~\eqref{eq_cont_sys_comp} can nevertheless be stated for our setup as follows.
\begin{prop}
	\label{prob_michel}
	Observe the DDS given by \eqref{eq_cont_sys_comp}. Assume that the unbounded discrete subset $\mathcal{T}^a$ of $\mathbb{R}_{\geq 0}$ satisfies 
	\begin{equation}
	\label{eq_non_mon_dec3}
	0 < \underline{\eta} \leq \tau^a_{l+1} - \tau^a_{l} \leq \overline{\eta},~ \forall l\in\mathbb{N}_0
	\end{equation}
	and $\tau^a_0 = 0$. Furthermore, assume there is a continuous positive definite function $V:\mathbb{R}^{2n_x} \rightarrow \mathbb{R},$ satisfying 
		\begin{equation}
		\label{eq_non_mon_dec4}
		\alpha_3 {(\norm{\xi})} \leq V(\xi) \leq \alpha_4 {(\norm{\xi})}, 	
		\end{equation}
	 such that for all $l\in \mathbb{N}_0,$ and all $\xi(\tau^a_l) \in \mathcal{X}_{c,2} $, where $\mathcal{X}_{c,2} := \left\lbrace \xi|V(\xi) \leq c \right\rbrace$ for the chosen $c \in \mathbb{R}_{>0}$,
	\begin{equation}
	\label{eq_non_mon_desc1}
	V(\xi(\tau^a_l+r)) \leq \alpha_5(V(\xi(\tau^a_l))), ~ 0 \leq r < \tau^a_{l+1} -\tau^a_l
	\end{equation}
	and
	\begin{equation}
	\label{eq_non_mon_desc2}
	\frac{1}{\tau^a_{l+1} - \tau^a_l} \left[ V(\xi(\tau^a_{l+1})) - V(\xi(\tau^a_l))\right] \leq -\gamma_2 (V(\xi(\tau^a_l)))
	\end{equation}
	hold with class $\mathcal{K}$ functions $\alpha_3, \alpha_4, \alpha_5, \gamma_2$. Then the equilibrium $\xi = 0$ is asymptotically stable for the DDS~\eqref{eq_cont_sys_comp}  with region of attraction $\mathcal{X}_{c,2}$. 
\end{prop}
\begin{proof}
	The proof is given in Appendix~\ref{append_a}.
\end{proof}
Based on  Assumption~\ref{asum_cont_clf} and our setup, we can simplify the conditions of Proposition~\ref{prob_michel} as follows.
\begin{prop}
	\label{prop_help_eq}
	Let Assumption~\ref{asum_cont_clf} hold. Then, \eqref{eq_non_mon_dec4}, \eqref{eq_non_mon_desc1} and \eqref{eq_non_mon_desc2} hold for the DDS~\eqref{eq_cont_sys_comp} and $V(\xi) = \frac{1}{2} \left(V_\CLF(\xi_1) + V_\CLF(\xi_2)\right)$, if \eqref{eq_non_mon_dec3} holds, and for all $x(\tau^a_l) \in \mathcal{X}_c$ and all $l \in \mathbb{N}_0$
	\begin{equation}
	\label{eq_v_desc}
	V_\CLF(x(\tau^a_l+\tone)) \leq V_\CLF(x(\tau^a_l)) 
	\end{equation} 
	holds for $0\leq \tone < \tau^a_{l+1}-\tau^a_l$,  
	 and
	\begin{align}
	\frac{1}{\tau^a_{l+1} - \tau^a_l} \left[ V_\CLF(x(\tau^a_{l+1})) - V_\CLF(x(\tau^a_l))\right]
	\leq -\gamma_2 (V_\CLF(x(\tau^a_l))) \label{eq_v_desc_2}
	\end{align}
	holds.
	Furthermore, $x(\tau^a_l) \in\mathcal{X}_c$ implies $\xi(\tau^a_l) \in\mathcal{X}_{c,2}$.
\end{prop}
\begin{proof}
	The proof is given in Appendix~\ref{append_c}.
\end{proof}
Proposition~\ref{prop_help_eq} allows us to state a sufficient stability condition in terms of the decrease conditions \eqref{eq_v_desc} and \eqref{eq_v_desc_2} for the system state $x$. This will turn out to be useful to guarantee stability for the \modelbasedPETC mechanism. %
Note that the choice of $h$ according to the $\sigma$-MASP bound from  \eqref{eq_h_max} guarantees satisfaction of the decrease conditions \eqref{eq_v_desc} and \eqref{eq_v_desc_2} if a transmission is triggered periodically at each sampling time. This is formalized in \cite{hertneck2019nonlinear}  by the following Lemma, which we will use in the stability proof for the PETC mechanism.
\begin{lemma}[cf. Lemma~1 from \cite{hertneck2019nonlinear}]
	\label{lem_cont_stab2}
	Let Assumptions~\ref{asum_cont_clf}-\ref{as_cont_pro3} hold for the chosen $c\in\mathbb{R}_{>0}$. Assume  system \eqref{eq_sys_cont} is used with controller \eqref{eq_cont_fb}, $\hat{x}(\tau^a_l) = x(\tau^a_l)$, and with $x(\tau^a_l)\in\mathcal{X}_c$ for some $\tau^a_{l}\in\mathcal{T}^a$. 
	If the next successful transmission takes place at a time $\tau^a_l+h$, i.e., $\tau^a_{l+1} = \tau^a_l+h$ and $ h  \leq h_{\sigma\text{\normalfont -MASP}}$ with $h_{\sigma\text{\normalfont -MASP}}$ according to \eqref{eq_h_max}, then
	\eqref{eq_v_desc} and \eqref{eq_v_desc_2} hold for $V_\CLF$  on $\mathcal{X}_{c}$
	with $\gamma_2 = \sigma  \gamma $. Moreover, if $V_\CLF(x(\tau^a_l)) \leq S(\tau^a_l,x_0)$, then \eqref{eq_cont_v_smaller} holds for $\tau^a_l \leq t \leq \tau^a_{l+1}$ and $V_\CLF(x({\tau^a_{l+1}})) \leq S(\tau^a_{l+1},x_0)$.
\end{lemma} 
\begin{proof}
	\remove{Follows from  the proof of Lemma~1 from } See \cite{hertneck2019nonlinear}.
\end{proof}
\begin{rema}
	Instead of Lemma~\ref{lem_cont_stab2}, different methods to determine a bound on the ($\sigma$-)MASP can be used in our setup. For example, the emulation approach from \cite{nesic2009explicit} can be used if a hybrid Lyapunov function is found for the considered system \eqref{eq_sys_cont}. Using this approach may for some setups lead to a larger bound on the ($\sigma$-)MASP, than the one we consider here, but requires some additional assumptions. 
\end{rema}

\section{Model-Based  PETC}
\label{sec_main}
\remove{In this section, we present the trigger mechanism for MB-PETC. %
We show that stability can be guaranteed if that mechanism is used, independently of the choice of $f_p$. 
Moreover, we demonstrate the efficiency of the proposed mechanism with a numerical example using an erroneous Euler forward discretization for prediction.}

In this section, we present the trigger mechanism for MB-PETC, provide stability guarantees that are independently of the choice of $f_p$ and demonstrate the efficiency of the proposed mechanism with a numerical example.

\subsection{Model-Based Triggering Mechanism}

\begin{table}[tb]
	\begin{algorithm}[H]
		\caption{MB-PETC triggering mechanism at $t = \tau_k^s$ for $k\in\mathbb{N}_0$.}
		\label{algo_cont_dyn_trig}
		\begin{algorithmic}[1]
			\IF{$k = 0$}
			\STATE $i^\text{\normalfont ref} \leftarrow 0$,  $\hat{x}^\text{\normalfont sens} \leftarrow x_0$,
			\STATE $V^\text{\normalfont ref} \leftarrow V_{\CLF}(x_0)$ with $V_{\CLF}$ according to Assumption~\ref{asum_cont_clf}
			\STATE send $x_0$ over the network
			\ELSE
			\STATE $\hat{x}^\text{\normalfont sens} \leftarrow f_p(\hat{x}^\text{\normalfont sens})$,
			\STATE $u^\text{\normalfont sens} \leftarrow \kappa(\hat{x}^\text{\normalfont sens})$
			\STATE $\lambda_k \leftarrow  h  \timeder{x(kh)}{u^\text{\normalfont sens}} +\frac{2}{3} h^{3/2} \mu_c $
			
			$\left( \norm{\partder{x(kh)}} \norm{f(x(kh),u^\text{\normalfont sens})} + \norm{f(x(kh),u^\text{\normalfont sens})}^2\right) $
			\IF {$k-i^\text{\normalfont ref} > \nu$ \textbf{or} $V_\CLF(\hat{x}^\text{\normalfont sens})>c$ \textbf{or}
				
				$V_\CLF(x(kh))+\lambda_k \geq V^\text{\normalfont ref} - (k-i^\text{\normalfont ref}+1)  h  \sigma\gamma \left( V^\text{\normalfont ref}\right)$} \label{line_trigger_cond} \label{line_trigger}
			\STATE send $x(kh)$ over the network 
			\STATE $i^\text{\normalfont ref} \leftarrow k$, $V^\text{\normalfont ref} \leftarrow V_{\CLF}(x(kh))$, $\hat{x}^\text{\normalfont sens} \leftarrow x(kh)$
			\ELSE
			\STATE no transmission of $x(kh)$ necessary
			\ENDIF

			\ENDIF
		\end{algorithmic}
	\end{algorithm}	\
	\vspace{-1.5 cm}
\end{table}

The \modelbasedPETC trigger mechanism is given by Algorithm~\ref{algo_cont_dyn_trig}. As in \cite{hertneck2019nonlinear}, an upper bound on the time evolution of the Lyapunov function is computed at each sampling time based on the bound from Corollary~\ref{coro_cont_lyap} and the current system state, and a transmission is triggered if this bound violates the convergence criterion. 

Due to the prediction at the actuator, the input changes at sampling times,
 such that the methods from \cite{hertneck2019nonlinear} cannot directly be used for the considered setup. Instead, the key idea here is, to use a copy of the prediction at the actuator also for the  PETC mechanism. On that way, it is possible 
to include knowledge about the input that will be applied to the system for the next sampling period in the triggering decision, and thus to bound the time evolution of the Lyapunov function for the next sampling period.

For technical reasons a transmission is  also triggered if the time between two transmissions exceeds the arbitrary large, but fixed bound $\nu h$, or if the predicted state $\hat{x}$ leaves the considered level set $\mathcal{X}_c$. We can state the following Theorem.
\begin{theo}
	\label{theo_cont_stab}
	Let Assumptions~\ref{asum_cont_clf}-\ref{as_cont_pro3} hold on $\mathcal{X}_c$ for the chosen $c \in \mathbb{R}_{>0}$. Assume  system \eqref{eq_sys_cont} is used with controller \eqref{eq_cont_fb}, prediction \eqref{eq_cont_est} and with $x_0\in \mathcal{X}_c$.  Let necessary transmissions be detected with the trigger mechanism specified by Algorithm~\ref{algo_cont_dyn_trig} that is evaluated periodically with a sampling period $ h \leq h_{\sigma\text{-\normalfont MASP}} $ with $h_{\sigma\text{-\normalfont MASP}}$ according to \eqref{eq_h_max}, $\sigma \in  \left(0, 1\right)$ and arbitrary large $\nu \in \mathbb{N}$. Then, the origin of the DDS~\eqref{eq_cont_sys_comp} is locally asymptotically stable with region of attraction $x_0\in\mathcal{X}_{c}$ and the convergence criterion \eqref{eq_cont_v_smaller} is satisfied. 
\end{theo}
\begin{proof}
	The proof is given in Appendix~\ref{append_b}.
\end{proof}
\begin{rema}
	A special case of Theorem~\ref{theo_cont_stab} for ZOH feedback is given by Theorem~1 from \cite{hertneck2019nonlinear}, if the maximum number of successive lost packets $m$ in Theorem~1 from \cite{hertneck2019nonlinear} is chosen to be $0$.
\end{rema}
\begin{rema}
	The computational complexity of Algorithm~\ref{algo_cont_dyn_trig} depends mainly on the computational complexity required for evaluating $f_p, \kappa,  \timeder{x(kh)}{u^\text{\normalfont sens}}, \norm{\partder{x(kh)}}$ and $\norm{f(x(kh),u^\text{\normalfont sens})}$  and therefore on the system dynamics \eqref{eq_sys_cont} and on the chosen prediction function and controller.
\end{rema}
\begin{rema}
	The main advantage of the proposed \modelbasedPETC mechanism, i.e., the prediction at the actuator, can also be combined with different PETC mechanisms from literature. For example, in the mechanism from \cite{Wang2016}, a transmission is triggered at time $kh$, if for a trigger function $\Upsilon$, the limit	 $\Upsilon(e(kh^-),x(kh^-))$ with $e(t) = \hat{x}(t) - x(t)$ is larger than $0$. The prediction at the actuator can here be included by considering for the trigger decision the error after the next prediction step is done, i.e., by using $\Upsilon(f_p(\hat{x}((k-1)h)-x(kh^-),x(kh^-))$ for the trigger decision.

\end{rema}

Theorem~\ref{theo_cont_stab} allows us to use a wide class of prediction functions at the actuator. As long as $f_p$ is continuously differentiable and satisfies $f_p(0) = 0$, stability is guaranteed independently of the actual choice of $f_p$. 
We will now show in a numerical example, that even a simple and erroneous prediction $f_p$ at the actuator can significantly improve the performance of the closed-loop system in comparison to the PETC with ZOH (subsequently denoted as ZOH-PETC). 
\subsection{Comparison of the Model-Based PETC to ZOH-PETC}
To compare the \modelbasedPETC mechanism to ZOH-PETC, we use the same inverted pendulum example, that was used in \cite{hertneck2019nonlinear}. For this example the system dynamics are given by \begin{equation}
\begin{pmatrix}
\dot{x}_1(t) \\
\dot{x}_2(t) 
\end{pmatrix} = \begin{pmatrix}
x_2(t)\\
(sin(x_1(t)) - u(t) cos(x_1(t))) \omega_0
\end{pmatrix}
\end{equation}
with pendulum angle $x_1$, angular velocity $x_2$, input $u$, that is a force that acts on the mass center of the pendulum, and a constant $\omega_0$. As in \cite{hertneck2019nonlinear}, we chose $\omega_0 = 0.1$, $\kappa(x) = \frac{31.6x_1+40.4x_2+sin(x_1)}{cos(x_1)}$ and $V_\CLF(x) = 1.278 x_1^2 + 0.632 x_1 x_2 + 0.404 x_2^2 $. For that choice, the DDS~\eqref{eq_cont_sys_comp} satisfies
for $c = 0.258$ and $\sigma = 0.35$ the assumptions of Theorem~\ref{theo_cont_stab} with 
$\left(\frac{3 (1-\sigma)}{2 \mu_c M_{\max,c}}\right)^2 = 2.77\cdot 10^{-5}$ and $(1+2L_1)^{-1} \approx \frac{1}{4.3}$. Thus, we chose $h= 2.77\cdot10^{-5} s$ as the sampling period length, for which Algorithm~\ref{algo_cont_dyn_trig} stabilizes the DDS~\eqref{eq_cont_sys_comp} independent of the used prediction function $f_p$. To demonstrate the proposed approach, we use a prediction that is based on an erroneous Euler forward integration. Thus, we use $f_p(\hat{x}) = \hat{x} +1.05hf(\hat{x},\kappa(\hat{x}))$. 

\begin{figure}[tb]
	\centering
	\centering
	\includegraphics[width = \linewidth]{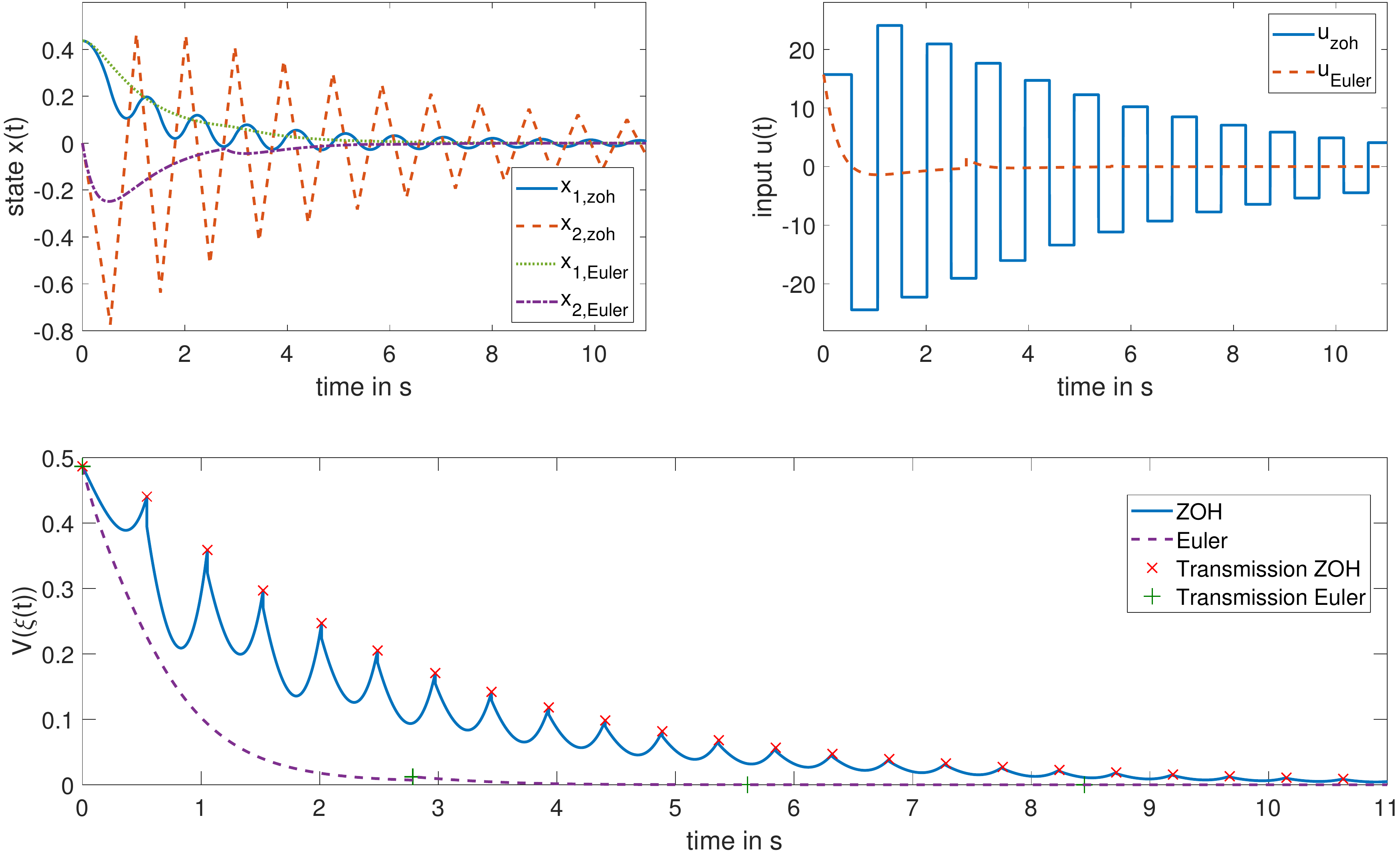}	\vspace{- 7mm}
	\caption{Comparison of ZOH-PETC (according to \cite{hertneck2019nonlinear}) and \modelbasedPETC with Euler forward prediction: System states (above left), inputs (above right) and Lyapunov functions (below) for the pendulum example.}
	\label{fig_example}
	\vspace{-7mm}
\end{figure}

In Figure~\ref{fig_example}, a comparison of the \modelbasedPETC mechanism with Euler forward based prediction to the ZOH-PETC mechanism from \cite{hertneck2019nonlinear} is given. Whilst the time evolution of the system states (above left) already heavily emphasizes that the \modelbased mechanism is superior for the considered example, the advantages become obvious from the input plot (above right) and from the Lyapunov function plot (below). The average input magnitude is significantly smaller and the Lyapunov function converges much faster for the \modelbasedPETC mechanism than for the ZOH-PETC mechanism. Nevertheless, significantly fewer transmissions are required for the \modelbasedPETC. Whilst the ZOH-PETC has an average time between two transmissions of $0.47 s$, there are at least $2.5 s$ between two transmissions of the \modelbasedPETC.

This emphasizes that the \modelbasedPETC mechanism can significantly improve the performance of the closed-loop system whilst reducing the amount of communication in comparison to the ZOH-PETC mechanism if a relatively simple and erroneous prediction is used even though the worst case guarantees are the same as for the ZOH-PETC mechanism.

\begin{rema}
For the considered example, the \modelbasedPETC mechanism triggers transmissions event tough the value of the Lyapunov function is much lower as it would be required to satisfy the convergence criterion (also satisfied by the trajectory for the ZOH-PETC mechanism). This is caused by the sufficient condition for the convergence criterion from Proposition~\ref{prop_cont_conv}, and can be omitted by using directly $S(t,x)$ for the trigger decision, if $S(t,x_0)$ is available. 
\end{rema}

\section{Conclusion}
\label{sec_conc}
In this paper, we presented an approach to combine PETC with MBNC for nonlinear systems. The proposed approach is based on predicting the behavior of the continuous-time plant at the actuator based on a discretization of the plant dynamics, if no current state information is received. Different approaches can be employed to implement the sampled-data prediction on simple hardware with limited computational capabilities. Even with such an inexact prediction, the \modelbased mechanism can reduce the required amount of communication significantly in comparison to prediction-free PETC approaches or to time-triggered control. Stability and a user-defined convergence speed are nevertheless guaranteed.

\vspace{-2mm}
\bibliography{sources}
\appendix
\subsection{Proof for Proposition~\ref{prob_michel}}
\label{append_a}
\begin{proof} This proof follows directly from the proof of Theorem 6.4.2 from \cite{michel2015stability}.
	We have to show that 
	\begin{equation}
	\label{eq_eps_del}
	\forall \epsilon> 0, \exists \delta(\epsilon) > 0, s.t. \norm{\xi_0} \leq \delta(\epsilon) \Rightarrow \norm{\xi(t)} \leq \epsilon~ \forall t \geq 0
	\end{equation}and
	\begin{equation}
	\label{eq_asym_conv}
	\norm{\xi(t)} \rightarrow 0~ \text{\normalfont as }t \rightarrow \infty
	\end{equation}  for all $\xi_0 \in \mathcal{X}_{c,2}$, to show asymptotic stability with region of attraction $\mathcal{X}_{c,2}$. Subsequently, we use the abbreviation $z_l := V(\xi(\tau^a_l))$. From \eqref{eq_non_mon_desc2}, we observe that $z_l \leq z_0$ for all $l \geq 0$ and $z_0 \leq c$. Using \eqref{eq_non_mon_dec4} and \eqref{eq_non_mon_desc1}, we obtain for $\tau^a_l \leq t <\tau^a_{l+1}$ that 
	\begin{equation}
	\label{eq_cont_z_bound}
	\norm{\xi(t)} \leq \alpha_3^{-1}\left(\alpha_5 \left(z_l\right)\right) 
	\end{equation}
	and thus $\norm{\xi(t)} \leq \alpha_3^{-1}\left(\alpha_5 \left(z_0\right)\right)$ for all $t \geq 0$.
	Hence, $\norm{\xi(t)} \leq \epsilon$ if $\norm{\xi_0} \leq \min \left\lbrace \alpha_4^{-1} \left(\alpha_5^{-1} \left(\alpha_3 \left( \epsilon \right)\right)\right), \alpha_4^{-1} (c) \right\rbrace := \delta(\epsilon)$ and thus \eqref{eq_eps_del} is satisfied. To show \eqref{eq_asym_conv}, we consider again \eqref{eq_non_mon_desc2} and observe by using \eqref{eq_non_mon_desc2} recursively that $z_{l+1} - z_{0} \leq - (\tau^a_{l+1} - 0) \gamma_2(z_l)$. With \eqref{eq_non_mon_dec3}, this yields $z_l \leq \gamma_2^{-1} \left( \frac{z_0}{l \underline{\eta}}\right)$. Consequently, condition \eqref{eq_asym_conv} follows directly due to \eqref{eq_cont_z_bound}.
\end{proof}
\subsection{Proof for Proposition~\ref{prop_help_eq}}
\label{append_c}
\begin{proof}
	Due to Assumption~\ref{asum_cont_clf}, \eqref{eq_non_mon_dec4} holds.  Next, observe that 
	\begin{equation}
	\xi_2(\tau^a_l+r) = \conc{f_p}{j}(\xi_2(\tau^a_l))
	\end{equation}
	for $jh \leq r < (j+1) h$ and all $j \in \left\lbrace 0,\dots, \floor*{\frac{\tau^a_{l+1} - \tau^a_l}{h}}-1\right\rbrace$.
	
	Due to the continuous differentiability of $f_p$, on each compact set $\mathcal{C}$ there is  a Lipschitz constant $L_\mathcal{C}$, such that $\norm{f_p(x)} \leq L_\mathcal{C}\norm{x}$ for all $x \in 
	\mathcal{C}$. Hence, for $x \in \mathcal{X}_{c}$ it holds that $\norm{\conc{f_p}{j}(x)} \leq \left(\Pi_{i=1}^j L_{\mathcal{C}_i}\right) \norm{x}$, where $L_{\mathcal{C}_v}$ is a Lipschitz constant for the set $\mathcal{C}_v$, that is recursively defined as $\mathcal{C}_v :=\left\lbrace x | \norm{x} \leq \left(\Pi_{i=1}^{v-1} L_{\mathcal{C}_i}\right) x_{\max} \right\rbrace$ with $x_{\max} = \underset{x \in \mathcal{X}_c}{\max} \norm{x}$ for $v \geq 2$ and $\mathcal{C}_1 = \mathcal{X}_c$.

	 We know from \eqref{eq_non_mon_dec3} that $\tau^a_{l+1} - \tau^a_l \leq \overline{\eta}.$ Let $j_{\max} := \floor*{\frac{\overline{\eta}}{h}}-1$. Thus it holds, that 
 	\begin{align}
 	\norm{\xi_2(\tau^a_l+r)} \leq&  \max\left\lbrace 1,\Pi_{i=1}^{j_{\max}} L_{\mathcal{C}_i} \right\rbrace    \norm{\xi_2(\tau^a_l)} \nonumber \\
 	\leq& \max \left\lbrace 1,\left(L_{\mathcal{C}_{j_{\max}}}\right)^{j_{\max}}\right\rbrace \norm{\xi_2(\tau^a_l)} . 	\label{eq_xi2_bound}
 	\end{align}
	Observe that due to \eqref{eq_cont_clf2}, it follows that $\norm{\xi_2(\tau^a_l)} \leq \alpha_1^{-1} \left( V_\CLF(\xi_2(\tau^a_l))\right)$ holds and $V_\CLF(\xi_2(\tau^a_l+r)) \leq \alpha_2 (\norm{\xi_2(\tau^a_l+r)})$ holds for all $r$. Furthermore, we have $\xi_1(t) = x(t)$ for all $t$ and  $\xi_2(\tau^a_{l+1}) = \xi_1(\tau^{a-}_{l+1}) = x(\tau^a_{l+1})$ due to the structure of the DDS~\eqref{eq_cont_sys_comp}. Thus, we can conclude from \eqref{eq_v_desc} and \eqref{eq_xi2_bound} that \eqref{eq_non_mon_desc1} holds with $\alpha_5(\cdot) := \max \left\lbrace (\cdot), \alpha_2\left(  \left(L_{\mathcal{C}_{j_{\max}}}\right)^{j_{\max}} \alpha_1 ^{-1}(\cdot)\right)\right\rbrace $. Moreover, \eqref{eq_non_mon_desc2} follows  due to the structure of the DDS from \eqref{eq_v_desc_2} and $\xi(\tau^a_l) \in \mathcal{X}_{c,2}$ follows due to the structure of the DDS from $x(\tau^a_l) \in \mathcal{X}_c$.
\end{proof}
\subsection{Proof for Theorem~\ref{theo_cont_stab}}
\label{append_b}

\begin{proof}
	This proof follows the same structure as the proof of Theorem~1 from \cite{hertneck2019nonlinear}. We will show that the conditions from Proposition~\ref{prob_michel} hold for the DDS~\eqref{eq_cont_sys_comp} and the Lyapunov function $V(\xi) = \frac{1}{2} \left(V_\CLF(\xi_1) +V_\CLF(\xi_2)\right)$. Moreover, we show that \eqref{eq_cont_v_smaller} holds for $\tau^a_l \leq t \leq \tau^a_{l+1}$ and $V_\CLF(x(\tau^a_{l+1})) \leq S(\tau^a_{l+1},x_0)$ if $V_\CLF(x(\tau^a_{l})) \leq S(\tau^a_{l},x_0)$. 
	
	In Algorithm~\ref{algo_cont_dyn_trig}, a transmission is triggered if the number of sampling periods between two successive transmissions exceeds the bound $\nu$. Since in addition, the time between two transmissions is lower bounded by the length of the sampling period $h$, we know that \eqref{eq_non_mon_dec3} holds. Thus, if \eqref{eq_v_desc} and \eqref{eq_v_desc_2} hold for all $l\in\mathbb{N}_0$, then \eqref{eq_non_mon_dec4}, \eqref{eq_non_mon_desc1} and \eqref{eq_non_mon_desc2}  hold due to Proposition~\ref{prop_help_eq} and the assumptions of the theorem. We consider now an arbitrary $l\in \mathbb{N}_0$ with $V_\CLF(x(\tau^a_{l})) \leq S(\tau^a_l,x_0)$.
	
	First assume that $\tau^a_{l+1}$ = $\tau^a_l + h$. In this case, we know from Lemma~\ref{lem_cont_stab2} that \eqref{eq_v_desc} and \eqref{eq_v_desc_2} hold for $\tau^a_l$ and $\tau^a_{l+1}$ with $\gamma_2 = \sigma \gamma$  and that  \eqref{eq_cont_v_smaller} holds for $\tau^a_l \leq t \leq \tau^a_{l+1}$ and $V_\CLF(x(\tau^a_{l+1})) \leq S(\tau^a_{l+1},x_0)$ if $V_\CLF(x(\tau^a_l)) \leq S(\tau^a_l,x_0)$. 
	
	Now assume that $\tau^a_{l+1} > \tau^a_l + h$. In this case, we can choose $\tilde{k}\geq 0$ and $p> 1$, such that $\tau^a_{l} = \tau^s_{\tilde{k}}$ and $\tau^a_{l+1} = \tau^s_{\tilde{k}+p}$.
	Henceforth, we consider an arbitrary $j \in \left\lbrace1 ,\dots,p-1 \right\rbrace$. We know from line~\ref{line_trigger} of Algorithm~\ref{algo_cont_dyn_trig} that at $k=\tilde{k}+j$ with $\tau^s_{\tilde{k}+j} = (\tilde{k}+j)h$
	\begin{align}
	&V_\CLF(x(\tau^s_{\tilde{k}+j})) + \lambda_{\tilde{k}+j}\nonumber \\
	 <& V_\CLF(x(\tau^a_l))-(\tau^s_{\tilde{k}+j}-\tau^a_l+h) \sigma \gamma(V_\CLF(x(\tau^a_l)))	\label{eq_algo_desc}
	\end{align}
	holds with $\lambda_{\tilde{k}+j}$ from Algorithm~\ref{algo_cont_dyn_trig}, and that $\hat{x}(\tau^s_{\tilde{k}+j}) \in 
	\mathcal{X}_c$, since otherwise, a transmission would have been triggered at time $\tau^s_{\tilde{k}+j}$.
	
	We use subsequently the abbreviation $u^*_{\tilde{k}+j} := \kappa(\hat{x}(\tau^s_{\tilde{k}+j}))$ for which it holds  that $u^*_{\tilde{k}+j} = u^\text{\normalfont sens}$ at $\tilde{k}+j = k $. We define the auxiliary function 
	\begingroup
	\begin{align}
	& V_\boundd(x(\tau^s_{\tilde{k}+j}),u^*_{\tilde{k}+j},\ttwo)  \nonumber\\
	:=&  V_\CLF(x(\tau^s_{\tilde{k}+j})) + \ttwo \timeder{x(\tau^s_{\tilde{k}+j})}{u^*_{\tilde{k}+j}} \nonumber \\
	&+ \frac{2}{3} \ttwo^{3/2} \mu_c \left( \norm{\partder{x(\tau^s_{\tilde{k}+j})}} \norm{f(x(\tau^s_{\tilde{k}+j}),u^*_{\tilde{k}+j})} \right. \nonumber \\
	&+ \left. \norm{f(x(\tau^s_{\tilde{k}+j}),u^*_{\tilde{k}+j})}^2\right). \nonumber
	\end{align}
	\endgroup
	For $\ttwo = h$, we obtain with \eqref{eq_algo_desc} and the definition of $\lambda_{\tilde{k}+j}$ from Algorithm~\ref{algo_cont_dyn_trig} 
	\begin{align}	
	&V_\boundd(x(\tau^s_{\tilde{k}+j}),u^*_{\tilde{k}+j},h) \nonumber \\
	 <&  V_\CLF(x(\tau^a_l))-(\tau^s_{\tilde{k}+j}-\tau^a_l+h) \sigma \gamma(V_\CLF(x(\tau^a_l))).\label{eq_cont_bound_eq}
	\end{align}

	The second derivative of $V_\boundd(x(\tau^s_{\tilde{k}+j}),u^*_{\tilde{k}+j},\ttwo)$ w.r.t. $m$ is positive for  $0 < \ttwo \leq  h $, and thus, $V_\boundd(x(\tau^s_{\tilde{k}+j}),u^*_{\tilde{k}+j},\ttwo)$ must have its maximum w.r.t. $\ttwo$ on the interval $0 \leq \ttwo \leq h $ either at $\ttwo = 0$ or at $\ttwo =  h $, i.e., 
	\begin{align}
	& V_\boundd(x(\tau^s_{\tilde{k}+j}),u^*_{\tilde{k}+j},\ttwo ) \nonumber \\
	\leq&\max \left\lbrace  V_\CLF(x(\tau^s_{\tilde{k}+j})), V_\boundd(x(\tau^s_{\tilde{k}+j}),u^*_{\tilde{k}+j},h)    \right\rbrace < c \label{eq_v_bound_r}.
	\end{align}

	Now, we consider an auxiliary system starting at time $\tau^s_{\tilde{k}+j} $ defined by $\dot{\tilde{x}} = f(\tilde{x}, u^*_{\tilde{k}+j}),~\tilde{x}(0)  = x(\tau^s_{\tilde{k}+j})$, for which a unique solution exists at least until $t^*$, and have $V_\CLF(\tilde{x}( \ttwo )) = V_\CLF(x(\tau^s_{\tilde{k}+j}+\ttwo)) $ for $0 \leq \ttwo \leq \tau^s_{\tilde{k}+j+1}-\tau^s_{\tilde{k}+j}$. Then by Corollary~\ref{coro_cont_lyap}, we obtain  for $0\leq \ttwo\leq \min \left\lbrace \tau^s_{\tilde{k}+j+1}-\tau^s_{\tilde{k}+j}, (1+2L_{1,c})^{-1},t^* \right\rbrace$ 
	\begin{align}
	&\mathcal{L}_fV(\tilde{x}(\ttwo),u^*_{\tilde{k}+j})\nonumber\\ 
	\leq& \mathcal{L}_fV(x(\tau^s_{\tilde{k}+j}),u^*_{\tilde{k}+j})  \nonumber	+	 \sqrt{\ttwo} \mu_c   \left(\norm{ V_\CLF^{'}(x(\tau^s_{\tilde{k}+j})}  \vphantom{ \norm{f(x(\tau^s_{\tilde{k}+j}),u^*_{\tilde{k}+j}}^2 } \right. \\ 
	&\left.\norm{f(x(\tau^s_{\tilde{k}+j}),u^*_{\tilde{k}+j}} 	+ \norm{f(x(\tau^s_{\tilde{k}+j}),u^*_{\tilde{k}+j}}^2 \right). \nonumber
	\end{align}
	A time integration yields
	\begin{equation}
	V_\CLF(x(\tau^s_{\tilde{k}+j}+\ttwo))  =  V_\CLF(\tilde{x}( \ttwo )) \leq V_\boundd(x(\tau^s_{\tilde{k}+j}),u^*_{\tilde{k}+j},\ttwo) \label{eq_cont_bound2}
	\end{equation}
	for $0\leq \ttwo\leq \min \left\lbrace \tau^s_{\tilde{k}+j+1}-\tau^s_{\tilde{k}+j}, (1+2L_{1,c})^{-1},t^* \right\rbrace$. We know from \eqref{eq_v_bound_r} and the choice of $h$ according to \eqref{eq_h_max} that 
	$ h \leq \min \left\lbrace \tau^s_{\tilde{k}+j+1}-\tau^s_{\tilde{k}+j}, (1+2L_{1,c})^{-1},t^* \right\rbrace$. Thus,  \eqref{eq_v_desc_2}  follows from \eqref{eq_cont_bound_eq} and \eqref{eq_cont_bound2} for $j = p-1$ with $\tau^a_{l+1} = \tau^s_{\tilde{k}+p}+h$ and $\gamma_2 = \sigma \gamma$. Moreover,  \eqref{eq_v_desc} is guaranteed for $0 \leq r \leq h$ due to Lemma~\ref{lem_cont_stab2} and for $jh \leq r \leq (j+1)h $ due to \eqref{eq_cont_bound_eq}, \eqref{eq_v_bound_r} and \eqref{eq_cont_bound2}. Since \eqref{eq_cont_bound_eq}, \eqref{eq_v_bound_r} and \eqref{eq_cont_bound2} hold for each $j \in \left\lbrace 1,\dots,p-1\right\rbrace$, and $\tau^a_{l+1} - \tau^a_l = hp$, \eqref{eq_v_desc} holds between $\tau^a_{l}$ and $\tau^a_{l+1}$.

	It remains to show that the convergence criterion \eqref{eq_cont_v_smaller} holds for $\tau^a_l$ and $\tau^a_{l+1}$ and that $V_\CLF(x(\tau^a_{l+1})) \leq S(\tau^a_{l+1},x_0)$ for the considered case, if $V_\CLF(x(\tau^a_{l})) \leq S(\tau^a_{l},x_0)$. For  $\tau^a_l \leq t \leq \tau^a_l+h$, \eqref{eq_cont_v_smaller} is directly implied by Lemma~\ref{lem_cont_stab2} and $V_\CLF(x(\tau^a_l))\leq S(\tau^a_{l},x_0)$. For  $\tau^s_{\tilde{k}+j} \leq t \leq \tau^s_{\tilde{k}+j}+h,$  we can use for each $j\in\left\lbrace 1 \dots,p-1\right\rbrace$ Proposition~\ref{prop_cont_conv} with $s = \tau^a_l$, $r = \tau^s_{\tilde{k}+j}-\tau^a_l+h$, $C_1 = V_\boundd(\tau^s_{\tilde{k}+j},u^*_{\tilde{k}+j},h)$ and $C_2 = V_\CLF(x(\tau^a_{l}))$ and \eqref{eq_algo_desc} and obtain that $V_\boundd(\tau^s_{\tilde{k}+j},u^*_{\tilde{k}+j},h) \leq S(\tau^s_{\tilde{k}+j}+h,x_0)$. Using additionally \eqref{eq_v_bound_r} and \eqref{eq_cont_bound2}, we see that  \eqref{eq_cont_v_smaller} holds for $\tau^s_{\tilde{k}+j} \leq t \leq \tau^s_{\tilde{k}+j}+h$ for each $j\in\left\lbrace 1,\dots,p-1 \right\rbrace$ and hence for $\tau^a_l \leq t \leq \tau^a_{l+1}$. Moreover, for $j=p-1$, we observe that $V_\CLF(x(\tau^a_{l+1})) \leq S(\tau^a_{l+1},x_0)$, for the considered case, which finishes this proof.	
\end{proof}

\end{document}